\documentclass[letterpaper, 10pt, conference]{ieeeconf}
\IEEEoverridecommandlockouts
\usepackage{cite}
\usepackage{amsmath,amssymb,amsfonts}

\usepackage{graphicx}
\usepackage{textcomp}
\usepackage{xcolor}
\usepackage{amsmath, amssymb, amsthm}
\usepackage{algorithm}
\usepackage{algpseudocode}
\usepackage{booktabs}
\usepackage{hyperref}
\newtheorem{lemma}{Lemma}
\newtheorem{proposition}{Proposition}

\newtheorem{theorem}{Theorem}

\newtheorem{remark}{Remark}
\newtheorem{assumption}{Assumption}
\newtheorem{definition}{Definition}
\begin{document}

% For draft editing
\newcommand{\Red}[1]{\textcolor{red}{[{#1}]}}

\title{\LARGE \bf
An Evolutionary Algorithm for Actuator-Sensor-Communication Co-Design in Distributed Control
}
\author{Pengyang Wu and Jing Shuang (Lisa) Li% <-this % stops a space
\thanks{At the time this manuscript was written, P.W. and J.S.L. were both with the Department of Electrical Engineering and Computer Science, University of Michigan, Ann Arbor, MI, 48109, USA.
        {\tt\small \{pengyanw, jslisali\}@umich.edu}.    }%
}

\maketitle

\begin{abstract}
This paper studies the co-design of actuators, sensors, and communication in the distributed setting, where a networked plant is partitioned into subsystems with sub-controllers interacting with other sub-controllers. The objective is to jointly minimize control cost (i.e., LQ cost) and material cost (i.e., number of actuators, sensors, and communication links used). We approach this using an evolutionary algorithm that selectively prunes a baseline dense LQR controller, and provide convergence and stability analyses. Our approach is validated in simulations; it substantially outperforms naive pruning (over 80\% in most cases), and performs similarly as greedy pruning but with much-improved scalability.
\end{abstract}

\section{Introduction \& Motivation}
Large-scale networked dynamical systems arise in many modern engineering applications, including power grids and traffic networks.
This has motivated research on \emph{distributed control}, where control actions are computed using local information and limited inter-agent communication \cite{9535262, anderson2019system, shin}; in turn, this motivates \emph{co-design} methodologies that jointly consider control performance, sensor and actuator placement, and communication structure. 
Co-design generally translates to nonlinear mixed-integer programs (NLMIPs) or other combinatorial problems; existing approaches often focus on special subsets of the co-design problem which can be reduced to more tractable forms \cite{Pequito2015, Jiang2016, Moothedath2019, Singh2021}.
However, this reduction eliminates large portions of the design space, which may contain more cost-effective solutions. 
Conversely, non-reductive methods such as greedy search require enumeration and evaluation over all possible sensors, actuators, and communication links, which becomes intractable as system size grows large.
In this paper, we propose an alternative approach using evolutionary algorithms (EA), which are heuristic optimization methods inspired by natural selection \cite{Whitley2001, DeJong2014}. 
We focus on the distributed linear-quadratic (LQ) setting: the problem setup is provided in Section \ref{sec:problem_setup}, the EA approach is described in Section \ref{sec:ea}, and convergence and stability analyses are provided in Sections \ref{sec:convergence} and \ref{sec:stability}. 
For unstable plants, we also devise an additional stability-promoting mechanism for the EA in Section \ref{sec:gersh_repair}.
Simulations validating our approach are shown in Section \ref{sec:simulations}.

\section{Problem setup} \label{sec:problem_setup}

Consider a discrete-time linear time-invariant system:
\begin{equation} \label{eq:plant}
x_{k+1} = A x_k + B u_k
\end{equation}
where $x_k \in \mathbb{R}^{N_x}$ is the state vector and $u_k \in \mathbb{R}^{N_u}$ is the control input.
The system contains $N$ interconnected subsystems.
State and control are partitioned into $[x]_i$ and $[u]_i$ for each subsystem $i$; system matrices $A$ and $B$ are partitioned into $[A]_{ij}$ and $[B]_{ij}$, which capture subsystem interactions. 
System topology is described by an unweighted directed graph $\mathcal{G}(\mathcal{V}, \mathcal{E})$, where vertex $v_i$ corresponds to subsystem $i$, and edge $e_{ij} \in \mathcal{E}$ exists whenever $[A]_{ij} \neq 0$ or $[B]_{ij} \neq 0$. 
Let $\mathcal{A}\in \{0,1\}^{N \times N}$ be the adjacency matrix for this graph, where by convention $\mathcal{A}_{ii} = 0$. 
For a state-feedback controller $u_k = Kx_k$, row $[K]_{i,:}$ is the sub-controller at subsystem $i$, and $[K]_{ij}$ represents information required by sub-controller $i$ from state block $j$. 
We can similarly build a controller adjacency matrix $\mathcal{K}(K)$, where $\mathcal{K}(K)_{ij} = 1$ wherever $[K]_{ij} \neq 0$, and $\mathcal{K}_{ii} = 0$ by convention. 

The number of actuators used by the controller is $N_a(K) := \|K\|_{1,0}$ (number of nonzero rows). 
Similarly, the number of sensors used by the controller is $N_s(K) := \|K\|_{0,1}$ (number of nonzero columns).
The number of inter-subsystem communication links used by the controller can be written as $N_c(K) = \|\mathcal{K}(K)\|_0$ (number of nonzeros).

The high-level co-design problem seeks a state feedback controller $K$ for plant \eqref{eq:plant} that jointly minimizes the control objective $J$ as well as material costs (weighted by positive scalar penalties $w_a$, $w_s$, and $w_c$) required by this controller:
\begin{equation} \label{eq:codesign_original}
\begin{aligned}
\min_{K} & \quad J(K) + w_a N_a(K) + w_s N_s(K) + w_c N_c(K) \\
\mathrm{s. t.} & \quad \eqref{eq:plant}, u_k = Kx_k
\end{aligned}
\end{equation}

In this paper, we design $K$ by pruning the linear quadratic regulator (LQR); this is inspired by previous work on the structure of dense LQR controllers \cite{shin}.
Specifically, given cost matrices $Q \succeq 0$ and $R \succ 0$, we synthesize an optimal LQR controller $K_\mathrm{d}$ (which is generally dense), then set select entries in $K_\mathrm{d}$ to zero. 
We now reparameterize \eqref{eq:codesign_original} to reflect this. 
First, define the optimization vector: 
\begin{equation}
\theta = [\,\ell, \;\mathbf{a}, \;\mathbf{s}\,]
\end{equation}
where link count $\ell \in [1, N_u N_x]$ is an integer determining the number of nonzero elements to keep from $K_\mathrm{d}$, and actuator mask $\mathbf{a} \in \{0, 1\}^{N_u}$ and sensor mask $\mathbf{s} \in \{0, 1\}^{N_x}$ are binary vectors representing actuator and sensor selection.
The pruned (i.e., sparsified) controller $K_\mathrm{s}(\theta)$ is derived from $K_\mathrm{d}$ by first pruning communications via $\Pi_\ell(K_\mathrm{d})$, which retains the $\ell$ largest-magnitude elements of $K_\mathrm{d}$ and sets the rest to zero, then pruning actuators and sensors via $\Pi_{\mathbf{a}, \mathbf{s}}(K_\mathrm{d})$, which sets $K_{\bar{\mathbf{a}},:} = 0$ and $K_{:,\bar{\mathbf{s}}} = 0$.

We consider control objective $J(K) := \frac{J_\mathrm{LQR}(K)}{J_\mathrm{LQR}(K_\mathrm{d})}$, i.e., the infinite-horizon LQ performance of controller $K$ for plant \eqref{eq:plant} relative to the optimal dense controller. 
%Cost $J_\mathrm{LQR}(K)$ is
%\begin{equation}
%    \begin{aligned}
%        J_{\mathrm{LQR}}(K) &= \mathbb{E}_{x_0 \sim \mathcal{N}(0, \Sigma)} \left[ \sum_{k=0}^{\infty} \left( x_k^\top Q x_k + u_k^\top R u_k \right) \right] \\
%        & \text{under dynamics } \eqref{eq:plant}, \quad u_k = Kx_k
%    \end{aligned}
%\end{equation}
%This objective is infinite when closed-loop $A+BK$ is unstable. When the closed-loop is stable, the objective can be evaluated by solving the discrete-time Lyapunov equation for the closed-loop system $A+BK$ to obtain a matrix $P$, then taking the trace of $P\Sigma$. This cost goes to infinity when the closed-loop system is unstable. 
Then, the co-design problem can be rewritten as:
\begin{equation} \label{eq:codesign_ours}
    \begin{aligned}
        \min_{\theta} & \quad J_\mathrm{EA}(\theta) \\ 
        \text{where} & \quad J_\mathrm{EA}(\theta) = \frac{J_\mathrm{LQR}(K_\mathrm{s}(\theta))}{J_\mathrm{LQR}(K_\mathrm{d})} + w_a N_a(K_\mathrm{s}(\theta)) \\ 
        & \quad \quad + w_s N_s(K_\mathrm{s}(\theta)) + w_c N_c(K_\mathrm{s}(\theta))
    \end{aligned}
\end{equation}

%\textbf{Notation.} $\|\cdot\|$ denotes the induced $2$-norm for matrices and the Euclidean norm for vectors; $\|\cdot\|_F$ denotes the Frobenius norm.
%\begin{equation}
%\label{cost_ea}
%\begin{aligned}
%J_{\text{EA}}(\theta) 
%&= \frac{J_{\mathrm{LQR}}(K_{\text{s}}(\theta))}{J_{\text{BM}}} \\
%&\quad + \left(
%w_c \|K_{\text{s}}\|_0
%+ w_a \|a\|_1
%+ w_s \|s\|_1
%\right).
%\end{aligned}
%\end{equation}
%where $w_c, w_s, w_a$ are positive constants representative of co-design penalty. The optimization objective is:
%\begin{equation}
%\begin{aligned}
%\min_{\theta} \quad & J_{\text{EA}}(\theta) \\
%\text{s.t.} \quad &  1 \leq \ell \leq N_u N_x, \\
%&\theta = [\,\ell, \;\mathbf{a}, \;\mathbf{s}\,], \\
%& a_j \in \{0, 1\}, \quad j = 1, \dots, N_u, \\
%& s_i \in \{0, 1\}, \quad i = 1, \dots, N_x, \\
%\text{where} \quad & K_\mathrm{s}(\theta) = \Pi_{\mathbf{a},\mathbf{s}} \left( \text{trunc}%(K_{\mathrm{d}}, \ell) \right).
%\end{aligned}
%\label{eq:optimization_problem}
%\end{equation}

\section{Evolutionary algorithm for co-design} \label{sec:ea}

We propose an evolutionary algorithm (EA) to solve \eqref{eq:codesign_ours}.
EAs work upon a population of $N_p$ individuals; individual $i$ is characterized by its gene (i.e., parameter) $\theta^i$. 
The population at generation $t$ is written as $\mathcal{P}_t := \{\theta^1, \theta^2, \ldots \theta^{N_p} \}$. 
Our population $\mathcal{P}_0$ is initialized with randomly generated individuals; for each individual, we set $\ell = \|K_{\mathrm{d}}\|_0$ and randomly draw elements in $\mathbf{a}$ and $\mathbf{s}$ from a Bernoulli distribution. 
Then, at each generation $t$, we evaluate the cost $J_\mathrm{EA}(\theta^i)$ of each individual $\theta^i$ in $\mathcal{P}_t$. 
The $n_e$ individuals with the lowest cost are directly carried over to the next generation's population $\mathcal{P}_{t+1}$. 
The remainder of the next generation's population are generated using the following standard operations:
\begin{enumerate}
    \item $\text{Selection}(\mathcal{P}_t, \tau)$ samples two distinct individuals (i.e., parents) from population $\mathcal{P}_t$. The probability of individual $i$ being chosen is:
    \begin{equation}
    P(i) = \frac{\exp(-J_\mathrm{EA}(\theta^i)/\tau)}{\sum_{\theta^j \in \mathcal{P}_t} \exp(-J_\mathrm{EA}(\theta^j)/\tau)}\end{equation}
    where $\tau$ is the temperature parameter. Individuals with lower cost are more likely to be chosen.

    \item $\text{Crossover}(\theta^{p_1}, \theta^{p_2}, p_c)$ takes genes from two distinct individuals and crossover probability $p_c$, and generates a new (i.e., child) gene. First, draw $k$ from $\mathcal{U}(1, N_\theta-1)$ where $N_\theta$ denotes the length of the gene vector. Then, generate the new gene $\theta^c$ as
    $\theta^c = [\theta^1_{1:k} \theta^2_{k+1:N_\theta}]$.

    \item $\text{Mutation}(\theta, p_m, d)$ takes gene $\theta$, mutation probability $p_m$ and range $d$, and generates a mutated gene $\theta^m$. Construct mutation vectors $\mathbf{m_a}$ and $\mathbf{m_s}$ by drawing each element from $\text{Bernoulli}(p_m)$. Then, construct mutation scalar $\delta$ which is drawn from $\text{Unif}\{-d, d\}$. The mutated gene is $\theta^m = [\mathrm{sat}(\ell_c + \delta), \mathbf{a} \oplus \mathbf{m_a}, \mathbf{s} \oplus \mathbf{m_s}]$, where $\mathrm{sat}(\cdot)$ clips the value to interval $[1, N_u N_x]$, and $\oplus$ denotes the element-wise XOR operation.
\end{enumerate}

This process repeats until the maximum number of generations $G_{\max}$ is reached. 
At this point, the gene of best-performing individual, denoted $\theta^\star$, is used as the EA solution to \eqref{eq:codesign_ours}. 
The overall EA is summarized in Algorithm \ref{alg:ea}. 
The complexity of the algorithm is $\mathcal{O}(G_{\max} N_p N_x^3)$; this is dominated by infinite-horizon LQ cost evaluations, which are generally cubic in $N_x$.
%analysis, which requires either eigenvalue computation or solving a Lyapunov equation for each individual in each generation, both of which have cubic complexity.
We demonstrate the efficacy of EA in simulations in Section \ref{sec:simulations}. Now, we discuss the convergence and stability properties of this algorithm.

\begin{algorithm}[t]
\caption{EA-Based Sparse LQR Controller Co-Design}
\label{alg:ea}
\begin{algorithmic}[1]
\Statex \textbf{Input:} Plant matrices $A, B$
\Statex \quad \quad Objective parameters $Q, R, w_c, w_a, w_s$
\Statex \quad \quad EA parameters $N_p, G_{\max}, p_c, p_m, n_e, \tau, d$ 
\State Compute optimal LQR controller $K_\mathrm{d}$
\State Randomly generate $\mathcal{P}_0 = \{\theta^i\}_{i=1}^{N_p}$
\State \textbf{for} $t = 1$ \textbf{to} $G_{\max}$\textbf{:} 
\State \quad \textbf{for} each $\theta^i \in \mathcal{P}_{t-1}$\textbf{:}
\State \quad \quad Evaluate $J_\mathrm{EA}(\theta^i)$
\State \quad \quad $\ell^i \gets \max\{\,r(i',j') : [K_\mathrm{s}(\theta^i)]_{i'j'} \neq 0\,\}$
\State $\mathcal{P}_{t} \gets$ $\{n_e$ lowest-cost individuals from $\mathcal{P}_{t-1}\}$
\State \textbf{for} $j = n_e + 1$ \textbf{to} $N_p$\textbf{:}
\State \quad $\theta^{p_1}, \theta^{p_2} \gets \text{Selection}(\mathcal{P}_{t-1}, \tau)$
\State \quad $\theta^c \gets \text{Crossover}(\theta^{p_1}, \theta^{p_2}, p_c)$
\State \quad $\theta^c \gets \text{Mutation}(\theta^c, p_m, d)$
\State \quad $\mathcal{P}_{t} \gets \mathcal{P}_{t} \cup \theta^c$
\State $\theta^\star \leftarrow \mathrm{argmin}_{\theta^i \in \mathcal{P}_t}J_\mathrm{EA}(\theta^i)$
\Statex \textbf{Return:} $K_\mathrm{s}(\theta^\star)$
\end{algorithmic}
\end{algorithm}

\section{Convergence of EA co-design} \label{sec:convergence}
We derive a convergence bound for Algorithm \ref{alg:ea} as the link count $\ell$ is updated (holding $\mathbf{a}, \mathbf{s}$ fixed), leveraging the exponential decay of $K_\mathrm{d}$ \cite{shin}. 
This analysis could be extended to actuator and sensor placement via analogous submodularity arguments, which would similarly yield lower bounds on convergence rate; we leave joint treatment of $\ell$, $\mathbf{a}$, and $\mathbf{s}$ to future work.
First, we recall relevant definitions and results from \cite{shin} that will be used in our analysis.

\begin{definition}
\label{def:stability}
Let $L > 0$, $\alpha \in [0,1)$. Matrix $\Phi$ is \emph{$(L,\alpha)$-stable} if
$\|\Phi^k\| \leq L\alpha^k \, \forall k \geq 0$; $(A,B)$ is
\emph{$(L,\alpha)$-stabilizable} if $\exists K,\|K\| \leq L$, s.t. $A+BK$ is
$(L,\alpha)$-stable; and $(A,C)$ is \emph{$(L,\alpha)$-detectable} if
$(A^\top,C^\top)$ is $(L,\alpha)$-stabilizable.
\end{definition}

\begin{assumption}
\label{as:system_params} For some $L > 1$ and $\alpha,\gamma \in (0,1)$:
(1) $\|A\|,\|B\|,\|Q\|,\|R\| \leq L$; (2) $R \succeq \gamma I$; (3) $(A,B)$ is
$(L,\alpha)$-stabilizable; (4) $(A,Q^{1/2})$ is
$(L,\alpha)$-detectable.
\end{assumption}

Note that Assumption \ref{as:system_params} is equivalent to stabilizability.

\begin{lemma} \label{lemm:decay} Let $d_{\mathcal{G}}(i,j)$ denote the number of
edges in the shortest path between nodes $i$ and $j$ of the system graph
$\mathcal{G} = (\mathcal{V},\mathcal{E})$. Then the optimal LQR gain
$K_\mathrm{d}$ satisfies
$\|K_{\mathrm{d},ij}\| \leq \Upsilon\rho^{\,d_{\mathcal{G}}(i,j)}$ for a constant
$\rho \in (0,1)$ and a constant $\Upsilon \geq 1$, whose values are given in
\cite{shin}.
\end{lemma}

\begin{lemma} \label{lemm:lipschitz}
On every sublevel set $\mathcal{S} := \{K : J_{\mathrm{LQR}}(K) \leq c\}$,  $ \exists \,L_J > 0$ such that
$\|\nabla_K J_{\mathrm{LQR}}(K_1) - \nabla_K J_{\mathrm{LQR}}(K_2)\|_F
\leq L_J\|K_1-K_2\|_F$ for all $K_1,K_2 \in \mathcal{S}$ \cite{fazel2018global}.
\end{lemma}

Since $K_\mathrm{d}$ minimizes $J_{\mathrm{LQR}}$,
$\nabla_K J_{\mathrm{LQR}}(K_\mathrm{d}) = 0$, so Lemma \ref{lemm:lipschitz}
gives $J_{\mathrm{LQR}}(K) - J_{\mathrm{LQR}}(K_\mathrm{d}) \leq
\tfrac{L_J}{2}\|K - K_\mathrm{d}\|_F^2 \quad \forall K \in \mathcal{S}$.
Also, our EA keeps the top $n_e$ individuals, so
$J_\mathrm{EA}(\theta_t^*) \geq J_\mathrm{EA}(\theta_{t+1}^*)$. We now begin our analysis.

\begin{definition} \label{def:ea_trajectory}
For each generation $t$, define
\begin{itemize}
    \item[(a)] \emph{Best individual} $\theta_t^* := \mathrm{argmin}_{\theta \in \mathcal{P}_t} J_\mathrm{EA}(\theta)$, with link count $\ell_t^*$ and controller $K_t := K_{\mathrm{s}}(\theta_t^*)$.
    \item[(b)] \emph{Link reduction} $X_t := \ell_{t-1}^* - \ell_t^*$, $|X_t| \leq d$.
    \item[(c)] \emph{Effective truncation distance} $h_t$: with simplified gene $\hat{\theta} := [\ell^*_t, \mathbf{1}^\top, \mathbf{1}^\top]$ and $\mathcal{E}_\mathcal{K}$ the edge set of $\mathcal{K}(K_\mathrm{s}(\hat{\theta}))$, $h_t$ is the largest integer such that every pair $(i,j)$ with $d_{\mathcal{G}}(i,j) \leq h_t$ lies in $\mathcal{E}_\mathcal{K}$.
\end{itemize}
\end{definition}

\begin{lemma} \label{lemm:h_crit}
Define the \emph{cost-rate function} $\Phi(h)$ and the \emph{critical truncation distance} $h^*$ as
\begin{equation}\label{eq:Phi}
  \Phi(h) := \frac{L_J\,\Upsilon^2\,\rho^{2h}}
    {J_{\mathrm{LQR}}(K_\mathrm{d})}
    \!\left(\sqrt{N_u N_x}+\tfrac{1}{2}\right),
\end{equation}
\begin{equation}\label{eq:h_star}
  h^* := \left\lfloor
    \frac{1}{2|\!\log\rho|}
    \log\frac{L_J\,\Upsilon^2\,(\sqrt{N_u N_x}+\tfrac{1}{2})}
    {w_c\,J_{\mathrm{LQR}}(K_\mathrm{d})}
  \right\rfloor .
\end{equation}
If $h_t > h^*$, then $\Phi(h_t) < w_c$, which means pruning links yields a
certified net decrease in $J_{\mathrm{EA}}$ (Proposition
\ref{prop:net_improve}). %Once $h_t \leq h^*$ the certificate no longer applies; this does not imply that further pruning fails to help.
\end{lemma}

\begin{lemma}
\label{lemm:one_step}
With $K_{t-1} := K_\mathrm{s}(\theta^*_{t-1})$ and $K_{t} := K_\mathrm{s}(\theta^*_{t})$,
if $X_t \geq 1$ then
\begin{equation}\label{eq:one_step}
  \frac{J_{\mathrm{LQR}}(K_t)-J_{\mathrm{LQR}}(K_{t-1})}
    {J_\mathrm{LQR}(K_\mathrm{d})}
  \;\leq\; \Phi(h_{t-1})\,X_t.
\end{equation}
\end{lemma}

\begin{proof}
By Lemma~\ref{lemm:lipschitz},
$J_{\mathrm{LQR}}(K_t)-J_{\mathrm{LQR}}(K_{t-1}) \leq
\langle \nabla_K J_{\mathrm{LQR}}(K_{t-1}), K_t\!-\!K_{t-1}\rangle
+ \tfrac{L_J}{2}\|K_t\!-\!K_{t-1}\|_F^2$.
Since $\mathrm{supp}(K_t)\subset\mathrm{supp}(K_{t-1})$, each of the $X_t$
removed entries satisfies $\|K_{\mathrm{d},ij}\| \leq \Upsilon\rho^{h_{t-1}}$ by
Lemma~\ref{lemm:decay}, so
\begin{equation}\label{eq:diff_bound}
  \|K_t-K_{t-1}\|_F \leq \sqrt{X_t}\;\Upsilon\rho^{h_{t-1}}.
\end{equation}
Since $\nabla_K J_{\mathrm{LQR}}(K_{\mathrm{d}})=0$, we have
$\|\nabla_K J_{\mathrm{LQR}}(K_{t-1})\|_F \leq L_J\|K_{t-1}-K_{\mathrm{d}}\|_F
\leq L_J\sqrt{N_uN_x}\,\Upsilon\rho^{h_{t-1}}$. Combining by Cauchy--Schwarz and
using $\sqrt{X_t}\leq X_t$ for $X_t\geq 1$ bounds the increase by
$L_J\Upsilon^2\rho^{2h_{t-1}}(\sqrt{N_uN_x}+\tfrac{1}{2})X_t$; dividing by
$J_{\mathrm{LQR}}(K_\mathrm{d})$ gives $\Phi(h_{t-1})X_t$.
\end{proof}

\begin{proposition}
\label{prop:net_improve}
If $X_t \geq 1$ and $h_{t-1}>h^*$, then
\begin{equation}\label{eq:net_improve}
  J_{\mathrm{EA}}(\theta_{t-1}^*)-J_{\mathrm{EA}}(\theta_t^*)
  \;\geq\; \bigl(w_c - \Phi(h_{t-1})\bigr)\,X_t \;>\;0 .
\end{equation}
\end{proposition}

\begin{proof}
$J_{\mathrm{EA}}(\theta_{t-1}^*)-J_{\mathrm{EA}}(\theta_t^*)
= w_c X_t - [J_{\mathrm{LQR}}(K_t)-J_{\mathrm{LQR}}(K_{t-1})]/J_\mathrm{LQR}(K_\mathrm{d})$,
whose second term is at most $\Phi(h_{t-1})X_t$ by Lemma~\ref{lemm:one_step};
and $h_{t-1}>h^*$ gives $\Phi(h_{t-1})<w_c$, which is how \eqref{eq:h_star}
is obtained from $\Phi(h)=w_c$.
\end{proof}

\begin{proposition}[Improvement probability]
\label{prop:p_imp}
Fix $t$, assume $h_{t-1}>h^*$. If $K_\mathrm{s}([\ell_{t-1}^*\!-\!1,
\mathbf{1},\mathbf{1}]) \in \mathcal{S}$, then each offspring strictly
improves $J_{\mathrm{EA}}$ with probability $\frac{(1-p_m)^{N_u+N_x}}{N_p\,(2d+1)}$.
\end{proposition}

\begin{proof}
Follow one explicit improvement path. Selection chooses $\theta_{t-1}^*$ as the
first parent, with probability at least $1/N_p$; crossover preserves
$\ell_{t-1}^*$; mutation draws $\delta = -1$, with probability $1/(2d\!+\!1)$ while leaving both masks unflipped, with probability $(1-p_m)^{N_u+N_x}$. These three stages are mutually independent; this gives the desired result.

%The resulting offspring has $X_t = 1$, so Proposition~\ref{prop:net_improve} applies. A sufficient condition for the hypothesis is $\sqrt{N_uN_x}\,\Upsilon\rho^{h(\ell_{t-1}^*-1)} < \sigma_{\mathrm{crit}}$, quantified in Section~\ref{sec:stability}(Theorem~\ref{thm:lyap-stab}).
\end{proof}

\begin{theorem} \label{thm:conv}
Let $\Phi(\cdot)$ and $h^*$ be as in Lemma \ref{lemm:h_crit}; define
\begin{equation}\label{eq:p_imp_1}
  p_{\mathrm{imp}}(t) := \frac{(1-p_m)^{N_u+N_x}}{N_p(2d+1)}
  \,\mathbf{1}\{h_{t-1}>h^*\},
\end{equation}
where $\mathbf{1}$ is the indicator function, and let
$P_{\mathrm{imp}}(t) := 1-(1-p_{\mathrm{imp}}(t))^{N_p-n_e}$. Then:

(i) At every $t$ satisfying the assumptions of
Proposition \ref{prop:p_imp},
\begin{equation}\label{eq:EA_step}
  \mathbb{E}\bigl[J_{\mathrm{EA}}(\theta_{t-1}^*)-J_{\mathrm{EA}}(\theta_t^*)\bigr]
  \geq \bigl(w_c-\Phi(h_{t-1})\bigr)P_{\mathrm{imp}}(t) > 0 .
\end{equation}

(ii) $J_{\mathrm{EA}}(\theta_t^*)$ is
non-increasing and bounded below by $0$, hence convergent; since the right-hand
side of \eqref{eq:EA_step} is strictly positive whenever $h_{t-1}>h^*$, only
finitely many generations can satisfy that condition.
\end{theorem}

\begin{proof}
Part (i) combines Lemma \ref{lemm:one_step} and Proposition
\ref{prop:net_improve}, which give the decrease $w_c-\Phi(h_{t-1})$ per improving
offspring, with Proposition \ref{prop:p_imp}, which lower-bounds its probability;
taking expectations over the $N_p-n_e$ offspring gives $P_{\mathrm{imp}}(t)$.
Part (ii) is the monotone convergence theorem applied to
$\{J_{\mathrm{EA}}(\theta_t^*)\}$, non-increasing by elitism and bounded below.
\end{proof}

Summing \eqref{eq:EA_step} over the first $T$ certified generations,
\begin{equation}\label{eq:EA_sum}
  \mathbb{E}\bigl[J_{\mathrm{EA}}(\theta_0^*)-J_{\mathrm{EA}}(\theta_T^*)\bigr]
  \;\geq\; \sum_{t=1}^{T}\bigl(w_c-\Phi(h_{t-1})\bigr)\,P_{\mathrm{imp}}(t).
\end{equation}

%Section \ref{sec:simulations} compares \eqref{eq:EA_step} with the measured per-generation improvement.

\section{Stability analysis for EA co-design} \label{sec:stability}
We now analyze the stability of plant \eqref{eq:plant} in closed loop with EA-encoded controllers $K_\mathrm{s}(\theta^i)$ for $i \in \mathcal{P}_t$.
We first construct a Lyapunov function for optimal LQR controller $K_\mathrm{d}$ and quantify its stability margin.
We then relate this to the difference between the dense controller and the sparsified controller, i.e., $\Delta K := K_{\mathrm{d}} - K_{\mathrm{s}}$, and analyze the effects of communication link and actuator/sensor sparsification.

By Lemma \ref{lemm:decay} and \cite[Theorem~A.7]{shin}, $A+BK_{\mathrm{d}}$ (i.e., the plant with the LQR controller) is $(\Upsilon,\rho)$-stable. A Lyapunov function $V(x) = x^\top V^* x$ for this system can be found by solving the discrete Lyapunov equation for matrix $V^*$. Additionally, the $(\Upsilon,\rho)$-stability of this system yields
\begin{equation}\label{eq:lyap-sandwich}
  \|x\|^{2} \;\le\; V(x) \;\le\; M\,\|x\|^{2},
  \qquad M := \frac{\Upsilon^{2}}{1 - \rho^{2}},
\end{equation}

Rearranging \eqref{eq:lyap-sandwich} gives the unit-decrease
identity
\begin{equation}\label{eq:lyap-decrease}
  V(x) - V(A+BK_{\mathrm{d}} x) = \|x\|^{2} \quad \forall x \in \mathbb{R}^{N_x}.
\end{equation}

%Then, the closed-loop with sparse matrix $K_{\mathrm{s}}$ can be written as 
%$A + B K_{\mathrm{s}}
%  = A+BK_{\mathrm{d}} - B \Delta K$.

\begin{theorem}
\label{thm:lyap-stab} Define the stability margin
\begin{equation}\label{eq:sigma-crit}
  \sigma_{\mathrm{crit}}
  := \frac{1 - \rho^{2}}{4\,\Upsilon^{2}\,L
       \bigl(L + 2\,\Upsilon\rho\bigr)},
\end{equation}
where $L$ is from Assumption~\ref{as:system_params}(a) and
$\Upsilon,\rho$ are from Lemma~\ref{lemm:decay}.
If $\sigma := \|K_{\mathrm{d}} - K_{\mathrm{s}}\|_F \le \sigma_{\mathrm{crit}}$, then the closed-loop system 
$A+BK_{\mathrm{s}}$ is $(\Omega,\beta)$-stable, where
\begin{equation}\label{eq:omega-beta}
  \beta :=
    \Bigl(1 - \frac{1-\rho^{2}}{2\,\Upsilon^{2}}\Bigr)^{\!1/2}
    \!\in (0,1),
  \qquad
  \Omega :=
    \Bigl(\frac{\Upsilon^{2}}{1-\rho^{2}}\Bigr)^{\!1/2}.
\end{equation}
\end{theorem}

\begin{proof}

$A+BK_{\mathrm{s}} = A+BK_{\mathrm{d}} - B\Delta K$. Substitute into $V(x) = x^\top V^* x$ and apply ~\eqref{eq:lyap-decrease} to yield
\begin{align}
  &V(A+BK_{\mathrm{s}}\,x) - V(x)
  \notag\\
  &= -\|x\|^{2}
     + \underbrace{(B\Delta K\,x)^{\!\top}
       V^{\star}(B\Delta K\,x)}_{T_{1}}
  \notag\\
  &\quad
     - \underbrace{2(A+BK_{\mathrm{d}}\,x)^{\!\top}
       V^{\star} B\Delta K\,x}_{T_{2}}.
  \label{eq:lyap-diff}
\end{align}
We now bound terms $T_1$ and $T_2$. For $T_1$, by~\eqref{eq:lyap-sandwich} and Assumption~1(a)
($\|B\| \le L$),
\[
  T_{1} \le M\,\|B\|^{2}\,\sigma^{2}\,\|x\|^{2}
       \le M\, L^{2}\,\sigma^{2}\,\|x\|^{2}.
\]
For $T_2$, since $\|A+BK_{\mathrm{d}}\| \le \Upsilon\rho$
(from $(\Upsilon,\rho)$-stability at time $k=1$) and
$\|V^{\star}\| \le M$,
\[
  |T_{2}| \le 2\,M\,\Upsilon\rho\,L\,\sigma\,\|x\|^{2}.
\]

Combining with~\eqref{eq:lyap-diff}:
\begin{equation}\label{eq:lyap-diff-bound}
  V(A+BK_{\mathrm{s}} x) - V(x)
  \le -\Bigl(1 - M\,L\,\sigma\,
       \bigl(L\sigma + 2\,\Upsilon\rho\bigr)\Bigr)\,\|x\|^{2}.
\end{equation}
When $\sigma \le \sigma_{\mathrm{crit}}$, we further bound these terms. First, notice 
\[
  M L \sigma_{\mathrm{crit}}
  = \frac{\Upsilon^{2}}{1-\rho^{2}} \cdot L \cdot
    \frac{1-\rho^{2}}{4\Upsilon^{2} L(L+2\Upsilon\rho)}
  = \frac{1}{4(L+2\Upsilon\rho)}.
\]
Then, the quadratic term in the bound of $T_1$ is bounded by
$M L^{2} \sigma_{\mathrm{crit}}^{2}
  = M L \sigma_{\mathrm{crit}} \cdot L\sigma_{\mathrm{crit}}
  = \frac{1-\rho^{2}}
         {16\,\Upsilon^{2}(L + 2\Upsilon\rho)^{2}}
  \le \frac{1}{16}$,
since $\Upsilon \ge 1$, $L \ge 1$, and $1-\rho^{2} \le 1$.
The cross-term in the bound of $T_2$ is bounded by  $2 M \Upsilon\rho\, L\, \sigma_{\mathrm{crit}}
  = \frac{\Upsilon\rho}{2(L + 2\Upsilon\rho)}
  \le \frac{1}{4}$,
where the inequality uses $L + 2\Upsilon\rho \ge 2\Upsilon\rho$.
%(with the convention $0/0 = 0$ when $\rho = 0$).
Therefore,
\[
  M L \sigma(L\sigma + 2\Upsilon\rho)
  \le M L^{2} \sigma_{\mathrm{crit}}^{2}
      + 2 M \Upsilon\rho\, L\, \sigma_{\mathrm{crit}} 
  = \tfrac{5}{16}
  < \tfrac{1}{2},
\]

and~\eqref{eq:lyap-diff-bound} becomes
\begin{equation}\label{eq:lyap-half}
  V(A+BK_{\mathrm{s}} x)
  \le V(x) - \bigl(1 - \tfrac{5}{16}\bigr)\,\|x\|^{2}
  \le V(x) - \tfrac{1}{2}\,\|x\|^{2}.
\end{equation}

Dividing by $V(x)$ ($x \ne 0$) and using
$\|x\|^{2}/V(x) \ge 1/M$:
\[
  \frac{V(A+BK_{\mathrm{s}} x)}{V(x)}
  \le 1 - \frac{1}{2M} = 1 - \frac{1-\rho^{2}}{2\,\Upsilon^{2}}
  = \beta^{2}.
\]
Iterating: $V(x_k) \le \beta^{2k}\,V(x_0)$.
Translating back via~\eqref{eq:lyap-sandwich}:
\[
  \|x_k\|^{2}
  \le V(x_k)
  \le \beta^{2k}\,M\,\|x_0\|^{2},
\]
so $\|x_k\| \le \Omega\,\beta^{k}\,\|x_0\|$
with $\Omega = M^{1/2}$.
Since $\Upsilon \ge 1$ and $\rho \in (0,1)$,
we have $\beta^{2} = 1 - (1-\rho^2)/(2\Upsilon^2) \in (0,1)$.
\end{proof}

%\begin{remark} \label{rem:sigma-crit} The stability margin $\sigma_{\mathrm{crit}}$ depends on $(\Upsilon,\rho)$ from Lemma~1 through ~\eqref{eq:sigma-crit}. As $\rho\to 1$ (slow spatial decay), the numerator $1-\rho^2\to 0$ while the denominator grows, so $\sigma_{\mathrm{crit}}\to 0$, i.e., the allowable controller perturbation vanishes. \end{remark}

\begin{remark}
\label{rem:perf-connection}
This Lyapunov function $V(x)$ also enables performance analysis. Since $V$ decreases geometrically along the closed-loop trajectory with
rate $\beta^{2}$, standard perturbation arguments yield a
sub-optimality gap of the form
\[
  J_{\mathrm{LQR}}(K_{\mathrm{s}})
  - J_{\mathrm{LQR}}(K_{\mathrm{d}})
  \;=\; O\!\left(\frac{\sigma \|\Sigma\|^{2}}{1-\beta^{2}}\right).
\]
%which is exponentially small in the link budget $\ell$ (through $\sigma$'s dependence on truncation), paralleling the near-optimality result for $\kappa$-distributed control in \cite[Theorem~4.2]{shin}. In Section\Red{TODO} the Gershgorin repair operator provides an alternative stability certificate for individuals that fail the Lyapunov condition; the two mechanisms are complementary.
\end{remark}

We now examine gain perturbation $\Delta K$ due to communication and actuator/sensor pruning:
%Sparse controller $K_\mathrm{s}(\theta)$ is constructed from $K_{\mathrm{d}}$ via pruning, as detailed in previous sections. we decompose the gain perturbation from effects related to communication link pruning and actuator/sensor selection:
\begin{equation}\label{eq:deltaK-decomp}
  \Delta K
  = \underbrace{K_{\mathrm{d}}
    - \Pi_{\ell}(K_{\mathrm{d}})}
    _{\Delta K_{\mathrm{comm}}}
  + \underbrace{\Pi_\ell(K_{\mathrm{d}})
    - \Pi_{\mathbf{a},\mathbf{s}}\!\bigl(
      \Pi_\ell(K_{\mathrm{d}})\bigr)}
    _{\Delta K_{\mathrm{as}}},
\end{equation}

\begin{theorem}
\label{thm:stab-trunc}
Consider pruned controller $K_{\mathrm{s}}(\theta)$ with effective truncation distance $h$.
Define
\begin{equation}\label{eq:h_stab}
  h_{\mathrm{stab}} := \min\!\Bigl\{h\geq 0 :
    \Upsilon\!\!
    \sqrt{\sum_{r>h} N_{\Delta}(r)\,
      \rho^{2r}}
    < \sigma_{\mathrm{crit}}\Bigr\},
\end{equation}
where $N_{\Delta}(r) :=|\{(i,j):d_{\mathcal{G}}(i,j)=r\}|$. Then, 

\emph{(i)} If $\mathbf{a}=\mathbf{1}$,
$\mathbf{s}=\mathbf{1}$ and $h\geq h_{\mathrm{stab}}$,
then $A+BK_{\mathrm{s}}$ is
$(\Omega,\beta)$-stable with
$\Omega,\beta$ as
in~\eqref{eq:omega-beta}.

\emph{(ii)} For general $\mathbf{a}, \mathbf{s}$,
$A+BK_{\mathrm{s}}$ is
$(\Omega,\beta)$-stable whenever
\begin{equation}\label{eq:stab-general}
  \|\Delta K_{\mathrm{comm}}\|_F
  + \|\Delta K_{\mathrm{as}}\|_F
  < \sigma_{\mathrm{crit}}.
\end{equation}

\end{theorem}

\begin{proof}
First, we bound $\|\Delta K_{\mathrm{comm}}\|_F$. Let $h$ be the effective truncation distance
of $P_{I_\ell}(K_{\mathrm{d}})$
as in Definition~\ref{def:ea_trajectory}(c).
By Lemma~\ref{lemm:decay}, every discarded
entry satisfies
$|K^{\star}_{ij}|
  \leq \Upsilon\rho^{d_{\mathcal{G}}(i,j)}$
with $d_{\mathcal{G}}(i,j)>h$, so
\begin{equation}\label{eq:trunc-bound}
  \|\Delta K_{\mathrm{comm}}\|_F^2
  \;\leq\;
  \Upsilon^2
  \!\sum_{r=h+1}^{\infty}
    N_{\Delta}(r)\,\rho^{2r},
\end{equation}
where
$N_{\Delta}(r)
  :=|\{(i,j):d_{\mathcal{G}}(i,j)=r\}|$.
For bounded-degree graphs, this sum is
dominated by $\rho^{2(h+1)}$, giving
$\|\Delta K_{\mathrm{comm}}\|_F
  \sim O(\Upsilon\,\rho^{h})$.
By~\eqref{eq:deltaK-decomp} and the
triangle inequality,
$\|K_{\mathrm{d}}-K_{\mathrm{s}}\|_F
  \leq \|\Delta K_{\mathrm{comm}}\|_F
  + \|\Delta K_{\mathrm{as}}\|_F$.
When $\mathbf{a}=\mathbf{1}$,
$\mathbf{s}=\mathbf{1}$,
$\Delta K_{\mathrm{as}}=0$ and
Lemma~\ref{lemm:decay}
with~\eqref{eq:trunc-bound} gives
$\|\Delta K_{\mathrm{comm}}\|_F
  < \sigma_{\mathrm{crit}}$
for $h\geq h_{\mathrm{stab}}$.
Applying Theorem~\ref{thm:lyap-stab}
yields~(i).
Part~(ii) follows identically from
Theorem~\ref{thm:lyap-stab}.
%The monotonicity holds because
%$\Pi_{a,s}$ can only zero out additional
%entries of $P_{I_\ell}(K_{\mathrm{d}})$.
\end{proof}
% \begin{corollary}[Offspring stability probability]
% \label{cor:ea-stab}
% At generation $t$, suppose the best
% individual has link count $\ell_t^*$
% with $h(\ell_t^*)\geq h_{\mathrm{stab}}$.
% Then any single offspring is
% $(\Omega,\beta)$-stable with
% probability at least
% \begin{equation}\label{eq:p_stab}
%   p_{\mathrm{stab}}(t)
%   \;\geq\;
%   \frac{[\ell_t^*
%     -\ell_{\mathrm{stab}}
%     +d+1]_+}{2d+1}\;
%   (1-p_m)^{N_u+N_x},
% \end{equation}
% where $\ell_{\mathrm{stab}}
%   := \min\{\ell:h(\ell)
%   \geq h_{\mathrm{stab}}\}$,
% $[\cdot]_+:=\max(\cdot,0)$,
% and $d$ is the mutation range.
% In particular, when
% $\ell_t^*\geq\ell_{\mathrm{stab}}+d$,
% the first factor equals $1$ and
% stability is limited only by the
% mask-preservation probability
% $(1-p_m)^{N_u+N_x}$.
% \end{corollary}

% \begin{proof}
% Offspring link count is
% $\ell_t^*+\delta$ with
% $\delta\sim\mathrm{Unif}\{-d,\ldots,d\}$.
% Stability requires
% $\ell\geq\ell_{\mathrm{stab}}$
% (Theorem~\ref{thm:stab-trunc}(i))
% and unflipped masks (probability
% $(1-p_m)^{N_u+N_x}$).
% The two events are independent and gives the bound.
%\end{proof}
% In general, when the open-loop plant $A$ is stable, this bound is satisfied by the majority of individuals in any given EA generation. However, when the open-loop plant $A$ is unstable, this bound is violated and a substantial amount of individuals in each generation become unstable; this motivates the following section.

\section{Modifications for unstable open-loop plants}\label{sec:gersh_repair}
A few unstable individuals are generally not detrimental to the EA, but when plant $A$ is unstable, many individuals may become unstable, rendering the EA ineffective.  
Here, we introduce a modification to address this. 
For each individual $i$ where $A+BK_{\mathrm{s}}(\theta^i)$ is unstable, instead of directly using controller $K_{\mathrm{s}}(\theta^i)$, we introduce a repaired controller $K^r_{\mathrm{s}}(\theta^i)$ which has the same sparsity as $K_{\mathrm{s}}(\theta^i)$ but with different numerical values, such that $A+BK^r_{\mathrm{s}}(\theta^i)$ is stable. 
%In this way, we are able to efficiently utilize controllers encoded by genes $\theta^i$ to continue minimizing EA cost \eqref{eq:codesign_ours}. 
%The final EA output will then also be modified using this repaired controller. 

We leverage the Gershgorin disk theorem. For a matrix $M\in\mathbb{R}^{n\times n}$, define the
\emph{Gershgorin row-sum}
\begin{equation}\label{eq:Ri}
  R_i(M) :=  \sum_{j=1}^{n} |M_{ij}|, \qquad i=1,\dots,n,
\end{equation}
and the \textit{Gershgorin radius}
$\bar{R}(M) := \max_{i} R_i(M)$.

\begin{lemma}[Gershgorin sufficient condition]\label{lem:gers}
If\, $\bar{R}(A+BK)<1$, then $A+BK$ is Schur stable.
\end{lemma}
\begin{proof}
By the Gershgorin disk theorem, every eigenvalue $\lambda$ of
$A_{\mathrm{cl}}=A+BK$ lies in at least one disk
$\mathcal{D}_i=\{z\in\mathbb{C}:|z-A_{\mathrm{cl},ii}|\le r_i\}$
where $r_i=\sum_{j\neq i}|A_{\mathrm{cl},ij}|$.  Hence
$|\lambda|\le |A_{\mathrm{cl},ii}|+r_i = R_i(A_{\mathrm{cl}})
\le \bar{R}(A+BK)< 1$.
\end{proof}

We will now use gradients (w.r.t. the condition in Lemma \ref{lem:gers}) to improve the stability of $K_\mathrm{s}(\theta)$.

\begin{proposition}[Convexity]\label{prop:convex}
For each $i$, the function $K\mapsto R_i(A+BK)$ is convex.
Consequently, the Gershgorin-stable set
$\mathcal{G} := \{K:\bar{R}(A+BK)<1\}$ is a convex (open) subset
of\/ $\mathbb{R}^{N_u\times N_x}$, and
$\mathcal{G}\cap\mathcal{K}_\theta$ is convex.
\end{proposition}
\begin{proof}
Each entry $A_{\mathrm{cl},ij}=A_{ij}+\sum_{u=1}^{n_u}B_{iu}K_{uj}$ is affine
in $K$, so $|A_{\mathrm{cl},ij}|$ is convex.
Then, each row sum $R_i=\sum_j|A_{\mathrm{cl},ij}|$ is a non-negative sum of convex functions and therefore convex; $\bar{R}=\max_i R_i$ is also convex since it is a pointwise maximum.
Sublevel set $\{K:\bar{R}(A+BK)<1\}$ is therefore convex, and intersecting with the affine subspace $\mathcal{K}_\theta$ preserves convexity.
\end{proof}

% ----------------------------------------------------------------
%  4.  SUBGRADIENT DERIVATION
% ----------------------------------------------------------------

\begin{proposition}\label{prop:subgrad}
Let $i^*=\arg\max_i R_i(A+BK)$.  A subgradient of\/ $\bar{R}(A+BK)$ with respect to
$K_{uj}$ is
\begin{equation}\label{eq:subgrad}
  g_{uj} \;=\; \mathrm{sign}\!\bigl(A_{i^*j} + (BK)_{i^*j}\bigr)
               B_{i^*u}\,.
\end{equation}
%Then, the projected subgradient is$\tilde{g}_{uj}$$=g_{uj}\,\mathbf{1}_{(u,j)\in\mathcal{S}}$.
\end{proposition}

\begin{proof}
Since $\bar{R}(A+BK)=\max_i R_i(K)$, a subgradient of $\bar{R}$ at $K$
can be taken as any subgradient of $R_{i^*}$ at $K$
\cite[Sec.~3.1.2]{shin}.  Now,
\[
  R_{i^*}(K) = \sum_{j=1}^{n_x} \bigl|A_{i^*j} + (BK)_{i^*j}\bigr|
  = \sum_{j=1}^{n_x} \bigl|A_{i^*j} + \textstyle\sum_u B_{i^*u}K_{uj}\bigr|.
\]
Each summand $\phi_j(K) :=|A_{i^*j}+\sum_u B_{i^*u}K_{uj}|$ is
the absolute value of an affine function $\mu_j(K)$.  When
$\mu_j(K)\neq 0$, $\phi_j$ is differentiable with
$\partial\phi_j/\partial K_{uj}=\mathrm{sign}(\mu_j)\cdot B_{i^*u}$.
When $\mu_j(K)=0$, any value in
$[-|B_{i^*u}|,\,|B_{i^*u}|]$ is a valid subgradient element.
%Taking $\mathrm{sign}(0)=0$ (as in the implementation) yields a valid choice.  
Summing over $j$ gives the subgradient of $R_{i^*}$; since each $K_{uj}$ appears only in the $j$-th summand, the subgradient with respect to $K_{uj}$ reduces to~\eqref{eq:subgrad}.

\end{proof}

Scalar subgradients $g_{uj}$ can be stacked together to form matrix subgradient $g$. 
Now, we must ensure that controllers constructed using this are no denser than the original $K_\mathrm{s}(\theta)$.
Let $\mathcal{K}_\theta:= \{K\in\mathbb{R}^{N_u\times N_x}: \mathrm{supp}(K)\subseteq \mathrm{supp}(K_\mathrm{s}(\theta)) \}$ denote controllers with the desired sparsity level, and let $\mathrm{proj}_{\mathcal{K}_\theta}(g)$ denote projections (via zeroing elements) into this set.
Then, given gene $\theta^i$ where $A+BK_{\mathrm{s}}(\theta^i)$ is unstable and some target row-sum $\rho^*<1$, we solve convex program
\begin{equation}\label{eq:feasibility}
  \text{find } K^r \in\mathcal{K}_{\theta} \quad
  \text{s.t.}\quad \bar{R}(A+BK^r) \;\le\; \rho^*
\end{equation}
via iteration
\begin{equation}\label{eq:subgrad_update}
  K^{r,(t+1)} = \mathrm{proj}_{\mathcal{K}_\theta}\!\Bigl[
    K^{r,(t)} - \eta_t\,\tilde{g}^{(t)}
  \Bigr]
\end{equation}
with Polyak step size
\begin{equation}\label{eq:polyak}
  \eta_t = \frac{\bar{R}(K^{r,(t)})-\rho^*}
               {\|\tilde{g}^{(t)}\|_F^2}\ .
\end{equation}

\begin{proposition} \label{prop:phase1_conv}
If $\mathcal{G}\cap\mathcal{K}_\theta\neq\emptyset$, i.e., a Gershgorin-stable controller exists on the given sparsity
pattern, then subgradient iteration
\eqref{eq:subgrad_update}--\eqref{eq:polyak} generates a sequence
$\{K^{r,(t)}\}$ satisfying
\begin{equation}\label{eq:phase1_rate}
  \min_{0\le\tau\le t}\; \bar{R}(K^{r,(\tau)}) - \rho^*
  \;\le\;
  \frac{\|K^{r,(0)}-K^*\|_F^2}{2\,\sum_{\tau=0}^{t}\eta_\tau}
  \;\xrightarrow{t\to\infty}\; 0,
\end{equation}
where $K^*\in\mathcal{G}\cap\mathcal{K}_\theta$ is some feasible point.  In particular, the iterates reach $\bar{R}<1$ in
finitely many steps, which by Lemma~\ref{lem:gers} guarantees Schur
stability.
\end{proposition}

\begin{proof}
This is a standard result for Polyak-step subgradient methods applied
to convex feasibility~\cite{fazel2018global}.  Let $K^*$ be
any point with $\bar{R}(K^*)\le\rho^*$.  By convexity of $\bar{R}$:
\[
  \bar{R}(K^{r,(t)})-\rho^*
  \;\le\;
  \bar{R}(K^{r,(t)})-\bar{R}(K^*)
  \;\le\;
  \bigl\langle \tilde{g}^{(t)},\; K^{r,(t)}-K^* \bigr\rangle_F\,.
\]
From the update rule and the non-expansiveness of $\Pi_{\mathcal{S}}$:
\begin{align*}
\|K^{r,(t+1)}-K^*\|_F^2
\le \|K^{r,(t)}-\eta_t\tilde{g}^{(t)}-K^*\|_F^2 \\
= \|K^{r,(t)}-K^*\|_F^2
   - 2\eta_t\big\langle\tilde{g}^{(t)},K^{r,(t)}-K^*\big\rangle_F \\
\qquad\quad + \eta_t^2\|\tilde{g}^{(t)}\|_F^2 .
\end{align*}
Substituting the Polyak step size $\eta_t =
(\bar{R}(K^{r,(t)})-\rho^*)/\|\tilde{g}^{(t)}\|_F^2$ and using
the subgradient inequality:
\[
  \|K^{r,(t+1)}-K^*\|_F^2
  \;\le\;
  \|K^{r,(t)}-K^*\|_F^2
  - \frac{(\bar{R}(K^{r,(t)})-\rho^*)^2}{\|\tilde{g}^{(t)}\|_F^2}\,.
\]
Therefore $\{\|K^{r,(t)}-K^*\|_F^2\}$ is nonincreasing and
\[
  \sum_{t=0}^{\infty}
  \frac{(\bar{R}(K^{r,(t)})-\rho^*)^2}{\|\tilde{g}^{(t)}\|_F^2}
  \;\le\; \|K^{r,(0)}-K^*\|_F^2 < \infty.
\]
Since $\|\tilde{g}^{(t)}\|_F$ is bounded (by
$\|B\|_F\sqrt{n_x}$), numerator
$(\bar{R}(K^{r,(t)})-\rho^*)^2\to 0$.
The rate~\eqref{eq:phase1_rate} follows from the standard telescoping
argument for subgradient methods.
\end{proof}

Thus, if a Gershgorin-stable controller exists on the given sparsity pattern, our method is guaranteed to find it. 
To use this, two changes to Algorithm \ref{alg:ea} are required.
First, whenever $J_\mathrm{EA}(\theta^i)$ evaluates to infinity (Line 5), run \eqref{eq:subgrad_update}--\eqref{eq:polyak} to obtain controller $K^r_\mathrm{s}(\theta^i)$ and use that (instead of $K_s(\theta^i)$) in the cost evaluation, or declare it nonexistent if iterations fail to converge within a fixed budget.
Second, return $K^r_\mathrm{s}(\theta^\star)$ instead of $K_\mathrm{s}(\theta^\star)$.

\section{Simulations} \label{sec:simulations}
We first demonstrate the efficacy of Algorithm \ref{alg:ea}. All experiments use $Q = I_{N_x}$, $R = I_{N_u}$, $\Sigma=I_{N_x}$, cost weights $w_c = 0.05$, $w_a = 0.4$, $w_s = 0.2$, and EA parameters $N_p = 20$, $G_{\max} = 150$, $p_c = 0.8$,
$p_m = 0.05$, $n_e = 10$, $\tau = 0$, $d = 5$. We test on stable, linearized swing equations embedded in randomized square grid topologies (similar to \cite{anderson2019system}) of various sizes, and a marginally stable version of the same equations embedded in the IEEE 13-bus topology. Code to reproduce simulations may be found at \texttt{github.com/pengyanw/EA}. All simulations presented in this section run on a standard laptop computer (at least, on the first author's laptop) in about 60 seconds.
%\begin{table}[h]
%\centering
%\caption{Simulation parameters for the three experiments.}
%\label{tab:sim_params}
%\begin{tabular}{lccc}
%\hline
%\textbf{Parameter} & \textbf{$5\times5$ Grid} & \textbf{$7\times7$ Grid}
%                   & \textbf{IEEE 13-bus} \\
%\hline
%$N_x$              & 50   & 98   & 26  \\
%$N_u$              & 25   & 49   & 13  \\
%$T_s$ (s)          & 0.2  & 0.2  & 0.2 \\
%Spec. radius of $A$& $<1$ & $<1$ & $=1$ \\
%Gershgorin repair  & off  & off  & off  \\
%$\kappa$-trunc ref & $\kappa=1$  & $\kappa=1$ & $\kappa=2$ (all stable) \\
%\hline
%\end{tabular}
%\end{table}

Results are shown in Figure ~\ref{fig:perf-conv}. 
We include comparisons to two baselines: original (dense) LQR controller $K_\mathrm{d}$ and diagonal LQR (i.e., dense LQR with all cross-subsystem communication links removed).\footnote{We also tested intermediate truncations of the type suggested in \cite{shin} but found that surprisingly, they are often outperformed by one of these two baselines. These are omitted for simplicity's sake.}
EA outperforms baselines by a wide margin (on average over 80\%).
The optimal EA-based controller and associated communication, actuator, and sensor selections is also shown in Figure \ref{fig:perf-conv}.%; we see that actuator selection is quite sparse for all three plants. 

\begin{figure*}
  \centering
  \begin{minipage}[b]{\textwidth}
    \centering
    \includegraphics[width=.68\linewidth]{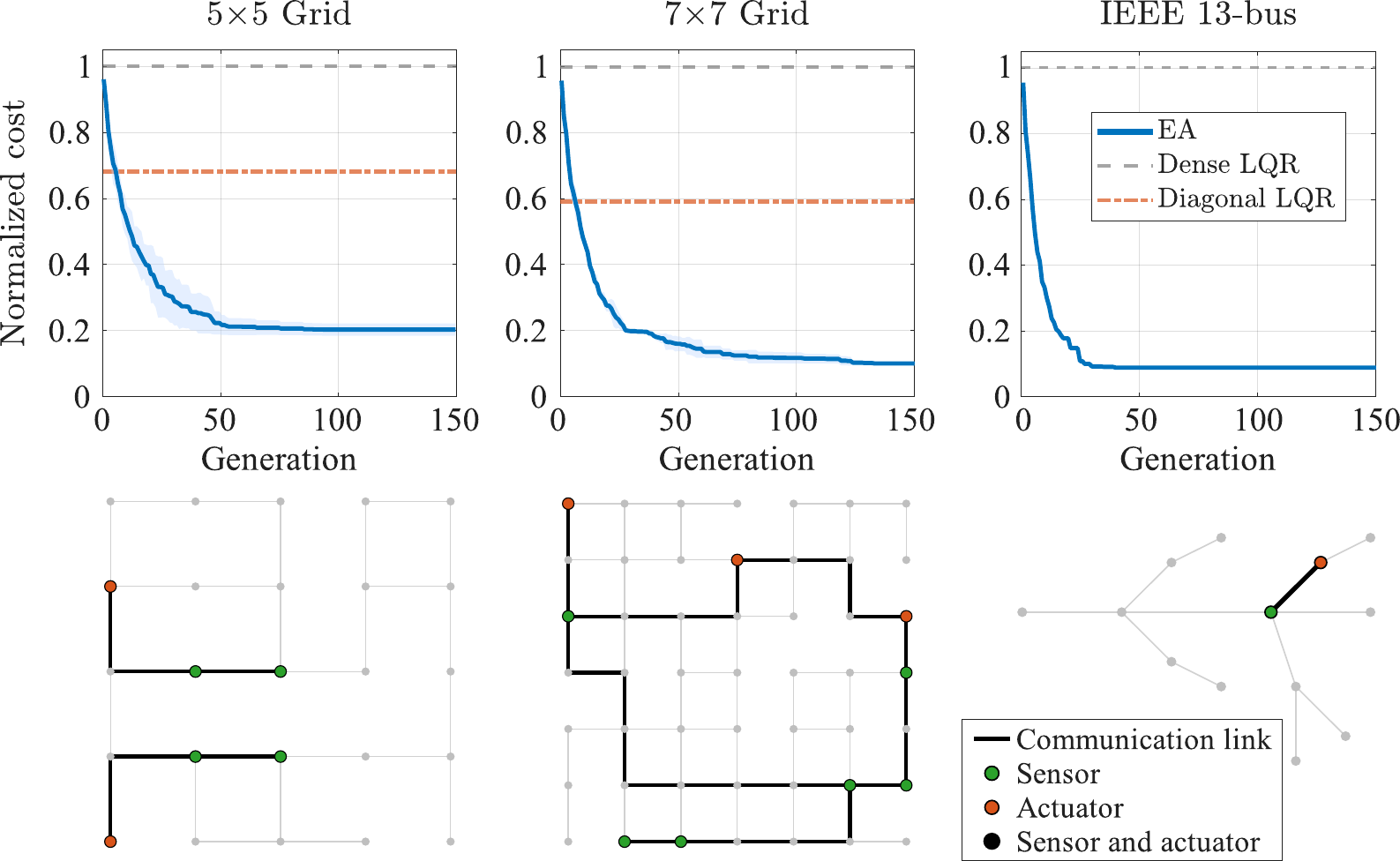}
    \caption{Results of running Algorithm \ref{alg:ea} on three different plants, averaged over 10 random seeds. Shaded regions indicate standard deviations. \textbf{Top panel:} normalized cost over generations; normalized cost is defined as $J_\mathrm{EA}(K^\star)/J_{\mathrm{EA}}(K_\mathrm{d})$, where $K^\star$ is the best per-generation controller. The solid blue line indicates EA performance; the dashed grey and red lines indicate the dense LQR and diagonal LQR baselines, respectively. For the IEEE 13-bus system, the diagonal LQR is unstable so it omitted. \textbf{Bottom panel:} Graphical depiction of one of the optimal controllers returned by the EA at termination and its link, actuator, and sensor selections. Grey circles and dashes indicate nodes and edges in the plant. Black edges indicate communication links used by the EA controller; green circles indicate sensors used by the EA controller; orange circles indicate actuators used by the EA controller.}
    \label{fig:perf-conv}
  \end{minipage}
\end{figure*}

We also compared our EA results with greedy pruning under the same setting and objective. In general, greedy pruning performs similarly as the EA; however, a key drawback of greedy pruning is that it requires many more cost evaluations. Specifically, in each generation, while EA only requires $N_p$ cost evaluations, the greedy algorithm requires a cost evaluation for every single possible communication link, actuator, and sensor removal. This effect is very pronounced for larger systems; for instance, for a 50 by 50 grid, the greedy algorithm requires 3 times as much computation to reach a similar performance to EA.
\begin{figure}
    \centering
    \includegraphics[width=0.4\textwidth]{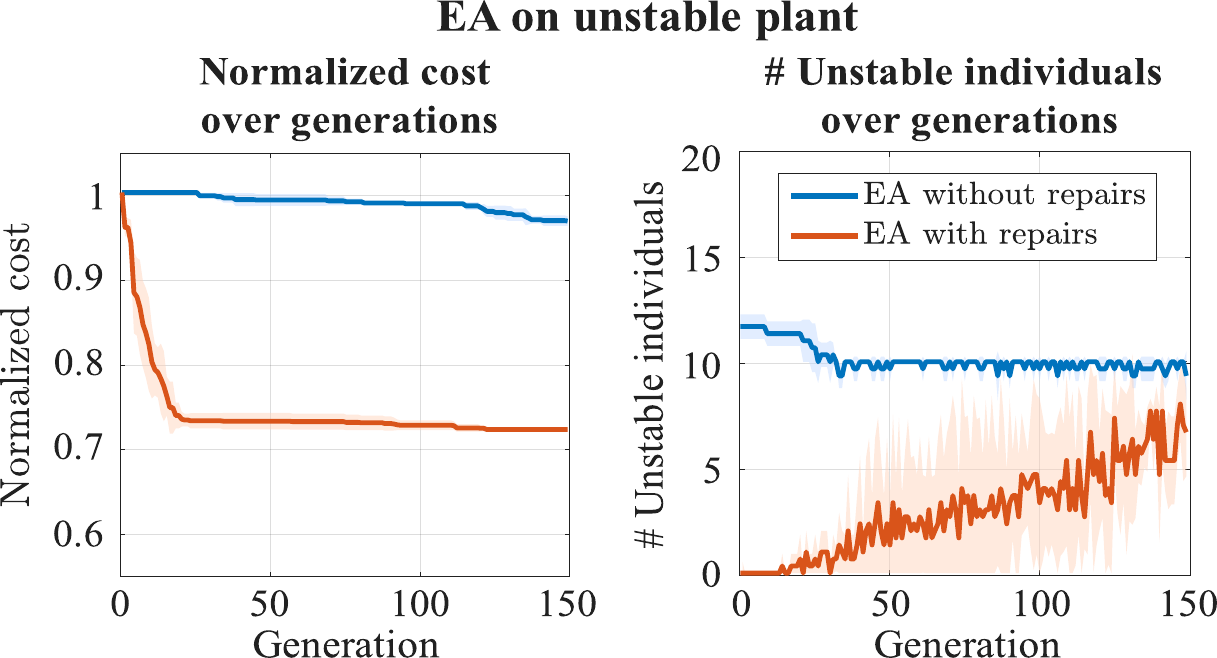}
    \caption{Results of running Algorithm \ref{alg:ea} with and without the repair mechanism from Section \ref{sec:gersh_repair} on an unstable plant, averaged over 3 random seeds. Shaded regions indicate standard deviations. EA with repairs far outperforms EA without repairs.}
    \label{fig:gers_comparison}
\end{figure}

%Both pay $\Theta(N_x^3)$ per evaluation, so only the evaluation count
% separates them, and only the EA's is fixed at $N_p G_{\max}$: a
% best-improvement sweep costs one evaluation per surviving mask bit and commits
% at most one removal, so Greedy pruning needs $\Theta(N_u^2 + N_x^2)$
% --- intrinsic to steepest descent.

Next, we demonstrate the effectiveness of our proposed repair mechanism on an unstable system. We use the same $5{\times}5$ grid as previously, but scale plant matrix $A$ so that it is unstable with spectral radius of $1.3$.
We compare the performance of EA without repair and EA with repair (as detailed in Section \ref{sec:gersh_repair}), with parameter $\rho^* = 0.95$. Without repair, EA only provides a 3\% improvement over the dense baseline. With repair, this improvement increases to 25\%. We also study the number of unstable individuals (as determined by $K_s(\theta^i)$ for the without-repair case and $K^r_\mathrm{s}(\theta^i)$ for the with-repair case). Without repairs, approximately half of the population is unstable at any given time. With repairs, early generations have nearly no unstable individuals. The number of unstable individuals rises in later generations, as the EA begins searching more and more sparse solutions that are more likely to be unrepairable.

% ------------------------------------------------------------------------

\section{Conclusions and future work}
In this paper, we proposed an evolutionary algorithm to perform co-design of LQ control cost and material cost (actuators, sensors, communication links) on a linear time-invariant plant, and demonstrated its efficacy in simulations. While this paper focuses on the LQ case, the general proposed EA framework and repair mechanism may be applicable to nonlinear systems as well; this will be the topic of future investigations.

\bibliography{cit}
\bibliographystyle{IEEEtran}

\end{document}